\documentclass{article}
\usepackage{amsmath,amsthm,amssymb}
\usepackage{algorithm}
\usepackage{algpseudocode}
\usepackage{array}      
\usepackage{booktabs}   
\usepackage{caption}
\usepackage{color}    
\usepackage{amsmath}
\usepackage{amssymb}
\usepackage{tabularx}
\usepackage{longtable}
\usepackage[latin1]{inputenc}

\newtheorem{Theorem}{Theorem}[section]
\newtheorem{Lemma}[Theorem]{Lemma}

\newtheorem{Corollary}[Theorem]{Corollary}

\newtheorem{Example}[Theorem]{Example}
\numberwithin{equation}{section}
\begin{document}
\title{{\LARGE Improved AntiGriesmer Bounds for Linear Anticodes and Applications}}

\author{Guanghui~Zhang$^1$, Bocong Chen$^2$, Liren Lin$^3$ and Hongwei Liu$^4$\footnote{E-mail addresses:
{\it zghui@squ.edu.cn (G. Zhang); mabcchen@scut.edu.cn (B. Chen); l\underline{~}r\underline{~}lin86@163.com (L. Lin);
hwliu@ccnu.edu.cn (H. Liu).}}
}

\date{\small
$1.$ School of Mathematics and Physics, Suqian University, Suqian 223800, China\\
$2.$ School of Mathematics, South China University of Technology, Guangzhou 510641, China\\
$3.$ School of Cyber Science and Technology, Hubei University, Wuhan 430062, China\\
$4.$ School of Mathematics and Statistics, Central China Normal University, Wuhan 430079, China
}


\maketitle

\begin{abstract}
This paper improves the antiGriesmer bound for linear anticodes previously established by Chen and Xie (Journal of Algebra, 673 (2025) 304-320).
While the original bound required the code length to satisfy $n < q^{k-1}$ and the dual code to have minimum distance at least 3,
our main result removes the length restriction and relaxes the dual distance condition to at least 2.
Specifically, we prove that for any $[n,k]_q$
linear anticode $\mathcal{C}$ over $\mathbb{F}_q$ with diameter $\delta$ and $d(\mathcal{C}^\perp) \geq 2$, the inequality
\[
n \leq \sum_{i=0}^{k-1} \left\lfloor \frac{\delta}{q^i} \right\rfloor
\]
holds. This generalization significantly broadens the applicability of the antiGriesmer bound.
We derive several corollaries, including lower bounds on the diameter $\delta$ in terms of $n$ and $k$, upper bounds on the code length $n$,
and constraints on the dimension $k$. Applications to the construction and classification of linear codes with few weights are also discussed,
along with examples demonstrating that our new bound can
be sharper than previous ones. Our work unifies and extends earlier findings,
providing a more comprehensive framework for studying linear anticodes and their properties.

	
	
\end{abstract}
	
\section{Introduction}
Let $q$ be a prime power and $\mathbb{F}_q$ the finite field with $q$ elements.
Let $n$ be a positive integer and $\mathbb{F}_q^n$ be the vector space of all $n$-tuples over the finite field $\mathbb{F}_q$.
Any nonempty subset $\mathcal{C}$ of $\mathbb{F}_q^n$ is called \emph{a code} over $\mathbb{F}_q$.
A \emph{linear code} $\mathcal{C}$ of length $n$ and dimension $k$ over $\mathbb{F}_q$, denoted as a linear $[n,k]_q$ code,
is a $k$-dimensional subspace of $\mathbb{F}_q^n$.
The \emph{Hamming weight} $\mathrm{wt}(\mathbf{c})$ of a codeword $\mathbf{c} \in \mathcal{C}$ is the number of nonzero coordinates in $\mathbf{c}$.
The \emph{Hamming distance} between two codewords
$\mathbf{c}_1, \mathbf{c}_2\in \mathcal{C}$ is the number of coordinates at
which $\mathbf{c}_1$ and $\mathbf{c}_2$ differ.
We denote the Hamming distance by $d(\mathbf{c}_1, \mathbf{c}_2)$.
Note that for all codewords $\mathbf{c}_1, \mathbf{c}_2\in \mathcal{C}$,
$d(\mathbf{c}_1, \mathbf{c}_2)=\mathrm{wt}(\mathbf{c_1}-\mathbf{c_2})$.
Since $\mathcal{C}$ is linear, the \emph{minimum distance} $d(\mathcal{C})$ of $\mathcal{C}$ can be defined as
\[
d(\mathcal{C}) = \min \{ \mathrm{wt}(\mathbf{c}) \mid \mathbf{c} \in \mathcal{C}, \mathbf{c} \neq \mathbf{0} \}
= \min \{ d(\mathbf{c}_1, \mathbf{c}_2) \mid \mathbf{c_1}, \mathbf{c_2} \in \mathcal{C} ~\text{and}~ \mathbf{c_1} \neq \mathbf{c_2} \}.
\]
A linear code $\mathcal{C}$ of dimension $k$ and minimum distance $d$ is denoted by a linear $[n,k,d]_q$ code.
For a linear $[n,k,d]_q$ code $\mathcal{C}$, the
\emph{diameter} $\delta(\mathcal{C})$ of
 $\mathcal{C}$ is defined as
\[
\delta(\mathcal{C}) = \max \{ \mathrm{wt}(\mathbf{c}) \mid \mathbf{c} \in \mathcal{C} \}.
\]
Based on the context, we abbreviate $\delta(\mathcal{C})$ as $\delta$ and $d(\mathcal{C})$ as $d$, which will not cause confusion.
When the diameter of a linear code is concerned,
it is often termed  as  a linear $[n,k,\delta]_q$ anticode.

The vector space $\mathbb{F}_q^n$ has a natural inner product defined on it.
In particular, if $\mathbf{x}=(x_1, x_2,\cdots, x_n), \mathbf{y}=(y_1, y_2, \cdots, y_n)$ are in $\mathbb{F}_q^n$,
we define the \emph{inner product} of $\mathbf{x}$ and $\mathbf{y}$ by
$$\langle\mathbf{x}, \mathbf{y}\rangle=x_1y_1+x_2y_2+\cdots+x_ny_n.$$
For a linear $[n,k,d]_q$ code $\mathcal{C}$,
the \emph{dual code} $\mathcal{C}^\perp$ of $\mathcal{C}$ is given by
\[
\mathcal{C}^\perp = \{ \mathbf{x} \in \mathbb{F}_q^n \mid \langle\mathbf{x}, \mathbf{c}\rangle
 = 0 \text{ for all } \mathbf{c} \in \mathcal{C} \}.
\]
Notice that $\mathcal{C}^\perp$ is also a linear code. The generator matrix for $\mathcal{C}^\perp$ is called \emph{the parity check matrix} for $\mathcal{C}$.
A linear code $\mathcal{C}$ is called \emph{projective} if the minimum distance of its dual code satisfies $d(\mathcal{C}^\perp) \geq 3$.

While the minimum distance receives  the most attention for
determining error-correction capability,
the diameter (maximum weight) also holds significant meaning for both theoretical and practical reasons.
For example, the diameter is directly linked to the covering radius problem. Studying codes with specific diameters is essential for finding optimal covering codes, which are vital for data compression \cite{CHLL}.
The existence of a codeword with weight equal to the diameter means there is a potential undetectable error pattern of that maximum size,
which directly influences the calculation of the probability of undetected error,
a key performance metric in communication systems designed for high-reliability error detection \cite[Part I, Chapter 1]{MS1977}.
In code-based cryptosystems like the McEliece cryptosystem, the security often relies on the hardness of problems related to finding codewords of a certain weight (e.g., the Syndrome Decoding Problem). The entire weight distribution, bounded by the diameter, affects the complexity of known attacks. Understanding the diameter and the number of high-weight codewords is crucial for selecting secure code parameters and assessing the cryptographic strength of the system \cite{OS}.

Previous works on the diameter of linear anticodes have focused on the following aspects:

(1) Bounds for anticodes:

\begin{itemize}
\item The code-anticode bound (\cite{AAK2001, Delsarte}):
 If  $\mathcal{C}\subset \mathbb{F}_q^n$ is a code with minimum distance $d$
and $\mathcal{A}\subset \mathbb{F}_q^n$ is an anticode with diameter $d-1$, then
$$|\mathcal{C}||\mathcal{A}|\leq q^n.$$
The sphere-packing bound \cite{huffman2003} can be thought
 of as a special case of the code-anticode bound, see \cite{AAK2001, Delsarte}.

\item The famous Erd$\ddot{o}$s-Kleitman bound for a binary anticode (\cite{EKR1961, kleitman}):
If  $\mathcal{C}$ is a binary anticode of length $n$ and
 diameter $\delta$, then
$$|\mathcal{C}|\leq \sum_{i=0}^{\lfloor\frac{\delta}{2}\rfloor}{n \choose i}.$$

\item The antiGriesmer   bound for projective linear anticodes \cite{CX}:
Let $q$ be a prime power and $n$ be a positive integer satisfying $n<q^{k-1}$.
Let $\mathcal{C}\subset \mathbb{F}_q^n$ be a projective linear anticode of dimension $k$, that is, $d(\mathcal{C}^\perp)\geq 3$.
Then its maximum weight (diameter) $\delta$
satisfies
\begin{equation}\label{knownresult}
n\leq \sum_{i=0}^{k-1}\left\lfloor\frac{\delta}{q^i}\right\rfloor.
\end{equation}

\end{itemize}
In addition, there is a lower bound $\delta \geq \frac{2^{k-1}n}{2^k-1}$ for a binary linear projective anticode of
dimension $k$ and diameter $\delta$ \cite{FARR1973} and a Gilbert-like bound on linear anticodes \cite{Reddy}.

(2) The constructions of linear anticodes and related codes from linear anticodes.
\begin{itemize}
\item There have been many constructions of large linear anticodes with a fixed diameter, see \cite{AK1998, EKR1961, FARR1973, kleitman}.

\item From a similar idea to that of constructing Solomon-Stiffler codes, Farrell gave a construction of linear codes with
optimal parameters or near optimal parameters from linear anticodes, see \cite{FARR1973} and \cite[pp. 547-556]{MS1977}.

\item Optimal locally repairable codes were constructed from anticodes in \cite{Silberstein}.

\item From known projective linear anticodes, simplex complementary codes with optimal or almost optimal
minimum distances were constructed,
and many new optimal or almost optimal few-weight linear codes, such as linear codes with two,
three and four nonzero weights, were constructed, see \cite{CX}.

\end{itemize}

As shown above, the study of bounds for linear anticodes is a fundamental and critical area in coding theory.
Understanding the structure and bounds of
such codes is essential for various problems in coding theory,
including the construction of error-correcting codes,
the analysis of code-anticode pairs, and applications in combinatorics and finite geometry.

In particular, in \cite{CX}, Chen and Xie introduced a new lower bound on the
diameter of projective linear anticodes over finite fields, known as the \emph{antiGriesmer bound}.
They also introduced the concept of  simplex complementary codes and
showed how to construct optimal or near-optimal few-weight linear codes from known projective linear anticodes.
This result strengthens the classical Erd$\ddot{o}$s-Kleitman bound for binary projective
linear anticodes and has led to new constructions of few-weight codes and related combinatorial objects.
Their antiGriesmer bound, however, requires the code length to satisfy $n < q^{k-1}$ and the dual distance to be at least 3,
see \cite[Theorem 2.2]{CX} or Equation (\ref{knownresult}).

In this paper, we improve the main result of \cite{CX} by removing the restriction
on the code length and relaxing the condition on the dual distance
(see  Theorem \ref{mainresult} and \cite[Theorem 2.2]{CX} for comparison).
We prove a new antiGriesmer-type bound that holds for all linear codes with dual distance at least 2,
without any assumption on the length, see Theorem \ref{mainresult} at the end of this section.
This not only leads to sharper constraints on the parameters of
linear anticodes and new applications in coding theory,
but also significantly broadens the applicability of the antiGriesmer bound and opens the door to new applications in the classification and construction of linear codes with few weights.

Furthermore, we derive several corollaries from our main bound, including lower bounds on the diameter $\delta$ in terms of $n$ and $k$, and upper bounds on the code length $n$ in terms of $\delta$ and the minimum distance. These results refine and extend the earlier findings, such as in \cite{CX} and \cite{FARR1973}, offering new insights into the structure of linear anticodes.

Now the main result of the paper can be stated as follows.

\begin{Theorem}[antiGriesmer Bound]\label{mainresult}
Let \(\mathcal{C}\) be an \([n, k]\) linear code over \(\mathbb{F}_q\) with diameter \(\delta\). If the minimum distance of the dual code \(\mathcal{C}^{\perp}\) is at least $2$, then
\begin{equation}\label{mainequation}
n \leq \sum_{i=0}^{k-1} \left\lfloor \frac{\delta}{q^i} \right\rfloor.
\end{equation}
\end{Theorem}
The proof of the above main result is provided in Section 3.
Combining with the well-known Griesmer bound (see \cite{Griesmer} or \cite[Theorem 2.7.4]{huffman2003}),
we have the following
upper and lower bounds on the code length $n$ (in the case where $k>1$
and $d(\mathcal{C}^{\perp})\geq2$)
$$
  \sum_{i=0}^{k-1} \left\lceil \frac{d}{q^i} \right\rceil\leq
  n \leq \sum_{i=0}^{k-1} \left\lfloor \frac{\delta}{q^i} \right\rfloor.
$$

\section{Preliminaries}
We first introduce some notation for later use. For a real number $x$,
the floor function, denoted as $\lfloor x\rfloor$, rounds the real number $x$ down to the nearest integer less than or equal to it,
i.e., $\lfloor x\rfloor$ is the largest integer that is not greater than $x$;
the ceiling function, denoted as $\lceil x\rceil$, rounds a real number up to the nearest integer greater than or equal to it,
i.e., $\lceil x\rceil$ is the smallest integer that is not less than $x$.
Let $S$ be a set, and denote the cardinality of $S$ by $|S|$;
let $S_0$ be a subset of $S$.
We denote the set of elements in $S$ but not in $S_0$
 by $S\backslash S_0$.

Let $\mathcal{C}$ be a linear $[n,k]_q$ code. A $k\times n$ matrix $G$ whose rows form a basis for $\mathcal{C}$
is called a \emph{generator matrix} for $\mathcal{C}$.
If $\mathcal{C}$ is an $[n,k]_q$ code with generator matrix $G$, then the codewords in $\mathcal{C}$ are
precisely the linear combinations of the rows of $G$.
Put another way,
$$\mathcal{C}=\{\mathbf{x}G\,|\,\mathbf{x}\in \mathbb{F}_q^k\},$$
where we regard $\mathbf{x}\in \mathbb{F}_q^k$ as a row vector.
Write $G=(\mathbf{g}_1, \mathbf{g}_2, \cdots, \mathbf{g}_n)$, that is, $\mathbf{g}_i$ is the $i$-th column of the matrix $G$,
where $i=1,2,\cdots,n$.
Then $\mathbf{x}\mathbf{g}_i$ is the $i$-th component of the codeword $\mathbf{x}G$ of $\mathcal{C}$.
Thus for any $\mathbf{x}\in \mathbb{F}_q^k$ the weight of the codeword $\mathbf{x}G$ is characterized by
$$\mathrm{wt}(\mathbf{x}G)=|\{i\,|\,\mathbf{x}\mathbf{g}_i\neq 0, 1\leq i\leq n\}|.$$

Let $\mathcal{C}$ be a linear $[n,k]_q$ code. As mentioned before,
the \emph{dual code} $\mathcal{C}^\perp$ of $\mathcal{C}$ is the set of all vectors in $\mathbb{F}_q^n$
that are orthogonal to every codeword in $\mathcal{C}$.
If $G$ is a generator matrix for $\mathcal{C}$, then
$$\mathcal{C}^\bot=\{\mathbf{x}\in \mathbb{F}_q^n|\mathbf{x}G^T=0\},$$
where $G^T$ denotes the transpose matrix of $G$.
Thus the dual code $\mathcal{C}^\bot$ of $\mathcal{C}$ is a linear $[n,n-k]_q$ code.
For any linear code $\mathcal{C}$, we have $(\mathcal{C}^\bot)^\bot=\mathcal{C}$.
Therefore, if $G$ and $H$ are generator and parity check matrices, respectively, for $\mathcal{C}$,
then $H$ and $G$ are generator and parity check matrices, respectively, for $\mathcal{C}^\bot$.

Let $\mathcal{C}$ be a linear $[n,k]_q$ code with the parity check matrix $H$.
Then $\mathcal{C}$ has minimum distance $d$ if and only if every $d-1$ columns of $H$ are linearly
independent and some $d$ columns are linearly dependent \cite[Theorem 10, Chapter 1]{MS1977}.
Therefore, we can reinterpret the projectivity of a linear code.
Recall that a linear code is called \emph{projective} if $d(\mathcal{C}^\perp) \geq 3$.
That is to say, a linear code $\mathcal{C}$ is called \emph{projective} if no two columns of its generator matrix are scalar multiples of each other,
or equivalently, any two columns of
the generator matrix of $\mathcal{C}$ are linearly independent.
For a linear code $\mathcal{C}$, if the minimum distance of the dual code $\mathcal{C}^\perp$ satisfies that $d(\mathcal{C^\perp})\geq 2$,
this means that the generator matrix $G$ for $\mathcal{C}$ has no zero column.

In the following, we review the definitions of the generalized Reed-Solomon code
and the extended generalized Reed-Solomon code, respectively, see \cite[Chapters 10 and 11]{MS1977}.

Let $\alpha_1, \alpha_2,\cdots,\alpha_n$ be distinct elements of $\mathbb{F}_q$
and $v_1,v_2,\cdots,v_n$ be nonzero elements of $\mathbb{F}_q$
(but not necessarily distinct).
Then the generalized Reed-Solomon code consists of all vectors
$$\big(v_1f(\alpha_1), v_2f(\alpha_2), \cdots,v_nf(\alpha_n)\big),$$
where $f(x)$ ranges over all polynomials of degree $<k$ with coefficients from $\mathbb{F}_q$.
Since $f(x)$ has at most $k-1$ zeros, the minimum distance is at least $n-k+1$, and hence is equal to $n-k+1$ from the Singleton bound \cite{huffman2003}.
Thus the generalized Reed-Solomon code is an $[n,k,n-k+1]_q$ linear code with the generator matrix as follows:
$$
\begin{pmatrix}
v_1 & v_2 & \cdots & v_n\\
v_1\alpha_1 & v_2\alpha_2 & \cdots & v_n\alpha_n\\
\vdots & \vdots & & \vdots\\
v_1\alpha_1^{k-1} & v_2\alpha_2^{k-1} & \cdots & v_n\alpha_n^{k-1}
\end{pmatrix}.
$$

The extended generalized Reed-Solomon code consists of all vectors
$$\big(v_1f(\alpha_1), v_2f(\alpha_2), \cdots,v_nf(\alpha_n), f_{k-1}\big),$$
where $f(x)$ ranges over all polynomials of degree $<k$ with coefficients from $\mathbb{F}_q$
and $f_{k-1}$ is the coefficient of $x^{k-1}$ in $f(x)$.
Then the extended generalized Reed-Solomon code is an $[n+1,k,n-k+2]_q$ linear code with the generator matrix as follows:
$$
\begin{pmatrix}
v_1 & v_2 & \cdots & v_n & 0\\
v_1\alpha_1 & v_2\alpha_2 & \cdots & v_n\alpha_n & 0 \\
\vdots & \vdots & \ddots & \vdots\\
v_1\alpha_1^{k-2} & v_2\alpha_2^{k-2} & \cdots & v_n\alpha_n^{k-2} & 0\\
v_1\alpha_1^{k-1} & v_2\alpha_2^{k-1} & \cdots & v_n\alpha_n^{k-1} & 1
\end{pmatrix}.
$$

Finally, we introduce a useful technique, the direct sum construction, that can be used to obtain new codes from old ones.
If $\mathcal{C}_1$ is a linear $[n_1,k_1,d_1]_q$ code  with
generator matrix $G_1$
and $\mathcal{C}_2$ is a linear $[n_2,k_2,d_2]_q$ code with
 generator matrix $G_2$,
then the linear code $\mathcal{C}$, denote   by $\mathcal{C}_1\oplus \mathcal{C}_2$, is the code
$$\mathcal{C}_1\oplus \mathcal{C}_2=\big\{(\mathbf{c}_1|\mathbf{c}_2)\,\big|\,\mathbf{c}_1\in \mathcal{C}_1, \mathbf{c}_2\in \mathcal{C}_2\big\}.$$
Clearly, $\mathcal{C}_1\oplus \mathcal{C}_2$ is a linear $[n_1+n_2, k_1+k_2, d=\min\{d_1,d_2\}]_q$ code with the generator matrix as follows:
$$\begin{pmatrix}
G_1 & 0\\
0 & G_2
\end{pmatrix}.
$$

\section{Main results and proofs}
Let $\mathcal{C}$ be an $[n,k,\delta]_q$ linear anticode over $\mathbb{F}_q$ and
$G$ be a $k\times n$ generator matrix for $\mathcal{C}$.
Let $\mathbf{g}_1,\mathbf{ g}_2,\cdots, \mathbf{g}_n$ be all the column vectors of $G$, that is,
$G=(\mathbf{g}_1, \mathbf{g}_2,\cdots,\mathbf{g}_n)$.
Assume that the minimum distance of the dual code \(\mathcal{C}^{\perp}\)
of the linear code $\mathcal{C}$ over $\mathbb{F}_q$ is at least $2$.
For $\mathbf{a}\in \mathbb{F}^k_q$, the $\mathbb{F}_q$-linear function $f_\mathbf{a}$ from $\mathbb{F}_q^k$ to $\mathbb{F}_q$ is defined to be
$$f_\mathbf{a}(\alpha)=\langle \mathbf{a}, \alpha \rangle, \forall~ \alpha\in \mathbb{F}_q^k.$$
For convenience, we simply denote this linear function by $\mathbf{a}$, i.e.,
$\mathbf{a}(\alpha)=f_\mathbf{a}(\alpha)= \langle \mathbf{a}, \alpha \rangle, \forall~ \alpha\in \mathbb{F}_q^k$. Denote by $(\mathbb{F}_q^k)^*$ the set of all the $\mathbb{F}_q$-linear functions from $\mathbb{F}_q^k$ to $\mathbb{F}_q$.

For code length $n$, we set $[n]=\{1,2,\cdots,n\}$.
Now we recursively construct linear functionals $\mathbf{a}_1, \mathbf{a}_2,\dots, \mathbf{a}_k \in (\mathbb{F}_q^k)^*$ and a partition of the index set
\[
[n] = S_1 \cup S_2 \cup \cdots \cup S_k,
\]
where some $S_i$ may be empty. Let $R_0=[n]$.

First, choose $\mathbf{a}_1 \neq 0$ maximizing $\mathrm{wt}(\mathbf{a}_1 G)$. Let
  \[
  S_1 = \{ j \in R_0 | \mathbf{a}_1(\mathbf{g}_j) \neq 0 \}, \quad |S_1| = \delta, \quad R_1 = [n] \setminus S_1.
  \]

Second, for $i \geq 2$, choose $\mathbf{a}_i \neq 0$ maximizing
  \[
  |S_i| = \left|\{ j \in R_{i-1} | \mathbf{a}_i(\mathbf{g}_j) \neq 0 \}\right|.
  \]

Define the remaining sets:
\[
\quad R_i = [n] \setminus (S_1 \cup \cdots \cup S_i).
\]

If $R_{i-1} = \emptyset$, set $S_i = \emptyset$.

\begin{Lemma}\label{linearlyindependent}
Keep the notation as above. If $R_{i-1}\neq \emptyset$ where $i\geq 2$, then $\mathbf{a}_1, \mathbf{a}_2, \cdots, \mathbf{a}_i$ are
linearly independent in the linear space
$(\mathbb{F}_q^k)^*$ over $\mathbb{F}_q$.
\end{Lemma}
\begin{proof}
Since $\mathbf{a}_1\neq 0$, $\mathbf{a}_1$ is linearly independent over $\mathbb{F}_q$.
Suppose that $\mathbf{a}_1, \mathbf{a}_2, \cdots, \mathbf{a}_{i-1}$ are linearly independent over $\mathbb{F}_q$.
In the following we need to prove that $\mathbf{a}_1, \mathbf{a}_2, \cdots, \mathbf{a}_{i}$ are linearly independent over $\mathbb{F}_q$.
Suppose that $\mathbf{a}_1, \mathbf{a}_2, \cdots, \mathbf{a}_{i}$ are linearly dependent over $\mathbb{F}_q$.
Then $\mathbf{a}_i$ can be expressed uniquely as a linear combination of $\mathbf{a}_1, \mathbf{a}_2, \cdots, \mathbf{a}_{i-1}$ with coefficients in $\mathbb{F}_q$:
$$\mathbf{a}_i=\lambda_1\mathbf{a}_1+\lambda_2\mathbf{a}_2+\cdots+\lambda_{i-1}\mathbf{a}_{i-1},$$
where $\lambda_j\in \mathbb{F}_q$ for $j=1,2,\cdots,i-1$.

From the definition of $R_{i-1}$ and the property $R_{i-1}\neq \emptyset$, we have the assertion as follows:
$$\mathbf{a}_1(\mathbf{g}_j)=\mathbf{a}_2(\mathbf{g}_j)=\cdots=\mathbf{a}_{i-1}(\mathbf{g}_j)=0, \forall ~j\in R_{i-1}.$$
If there exists an integer $m$ with $1\leq m\leq i-1$ such that $\mathbf{a}_m(\mathbf{g}_j)\neq 0$,
then from $j\in R_{i-1}\subseteq R_{m-1}$ we get that $j\in S_m$.
Thus, by the definition $R_{i-1}=[n]\setminus (S_1\cup\cdots \cup S_m\cup \cdots \cup S_{i-1})$ we must have $j \not\in R_{i-1}$,
which contradicts the condition $j\in R_{i-1}$. So the proof of the assertion is complete.

Now we calculate
$$\mathbf{a}_i(g_j)=\lambda_1\mathbf{a}_1(g_j)+\lambda_2\mathbf{a}_2(g_j)+\cdots+\lambda_{i-1}\mathbf{a}_{i-1}(g_j)=0, \forall ~j\in R_{i-1},$$
which shows that $S_i=\emptyset$ and thus we get that $|S_i|=0$.

On the other hand, according to the assumption that $R_{i-1}\neq \emptyset$ there is an integer $s\in R_{i-1}$.
Since the minimum distance of the dual code \(C^{\perp}\) is at least $2$, we have $\mathbf{g}_s\neq 0$.
Thus there exists a linear function $\mathbf{b}\in (\mathbb{F}_q^k)^*$ such that $\mathbf{b}(\mathbf{g}_s)\neq 0$.
By the definition of the set $S_i$ we obtain that
$$|S_i|\geq \left|\{j\in R_{i-1}|\mathbf{b}(\mathbf{g}_j)\neq0\}\right|\geq 1,$$
which contradicts the previous conclusion that $|S_i|=0$.

This completes the proof.
\end{proof}

\begin{Lemma}\label{partition}
Keep the notation as above. Then there is a disjoint partition of the column indices
$$[n]=S_1\cup S_2\cup \cdots \cup S_k.$$
Especially, we have
$$n=|S_1|+ |S_2|+ \cdots + |S_k|.$$
\end{Lemma}
\begin{proof}
According to Lemma \ref{linearlyindependent} and $\dim_{\mathbb{F}_q}(\mathbb{F}_q^k)^*=k$, we must have $R_{k}=\emptyset$,
which implies that $$[n]=S_1\cup S_2\cup \cdots \cup S_k.$$

To show that $S_i\cap S_j=\emptyset$ for $1\leq i\neq j\leq k$, suppose that $i<j$.
From the definition of $R_i$ and $S_i$, we have that
$$S_j\subseteq R_{j-1}\subset R_i=[n]\setminus  (S_1\cup S_2\cup \cdots \cup S_i),$$
and the result follows.
\end{proof}

\begin{Lemma}\label{inequality-1}
With the notation as above. For $i \geq 2$, we have
\[
|S_i| \leq \left\lfloor \frac{|S_{i-1}|}{q} \right\rfloor.
\]
\end{Lemma}
\begin{proof}
Let $\mathbb{P}(\mathbb{F}_q)=\mathbb{F}_q\cup \{\infty\}$ and consider the projective one-parameter family of linear functionals
\[
\mathbf{b}_t \in (\mathbb{F}_q^k)^*, \quad t \in \mathbb{P}(\mathbb{F}_q), \quad \mathbf{b}_t =
\begin{cases}
\mathbf{a}_{i-1} + t \, \mathbf{a}_i, & t \in \mathbb{F}_q, \\
\mathbf{a}_i, & t = \infty.
\end{cases}
\]
For $1\leq j\leq n$ and $t\in \mathbb{P}(\mathbb{F}_q)$,
let the symbol $\mathbf{1}_{\{\mathbf{b}_t(\mathbf{g}_j)\neq 0\}}$ be the indicator function (or characteristic function).
It means that it takes the value $1$ if the condition inside the curly braces is true, i.e., if $b_t(g_j)\neq 0$;
it takes the value $0$ if the condition is false, i.e., if $\mathbf{b}_t(\mathbf{g}_j)= 0$. That is to say,
$$\mathbf{1}_{\{\mathbf{b}_t(\mathbf{g}_j)\neq 0\}}=
\begin{cases}
1, & \mathbf{b}_t(\mathbf{g}_j)\neq 0, \\
0, & \mathbf{b}_t(\mathbf{g}_j)= 0.
\end{cases}
$$
Let $t$ run through $\mathbb{P}(\mathbb{F}_q)$ and consider
$$\sum_{t\in \mathbb{P}(\mathbb{F}_q)}\mathbf{1}_{\{\mathbf{b}_t(\mathbf{g}_j)\neq 0\}},$$
so we will proceed with a counting argument.
If $\mathbf{a}_i(\mathbf{g}_j)\neq 0$, then in this case we have
$$\mathbf{b}_t(\mathbf{g}_j)=0\Leftrightarrow \begin{cases}
\mathbf{a}_{i-1}(\mathbf{g}_j)+t\mathbf{a}_i(\mathbf{g}_j)=0, & t\in \mathbb{F}_q,\\
\mathbf{a}_i(\mathbf{g}_j)=0, & t=\infty.
\end{cases}
$$
Thus $\mathbf{b}_t(\mathbf{g}_j)=0$ if and only if $t=-\frac{\mathbf{a}_{i-1}(\mathbf{g}_j)}{\mathbf{a}_i(\mathbf{g}_j)}$,
which shows that there is a  unique $t$ such that $\mathbf{b}_t(\mathbf{g}_j)=0$.
If $\mathbf{a}_i(\mathbf{g}_j)=0$ and $\mathbf{a}_{i-1}(\mathbf{g}_j)\neq 0$, then
$$\mathbf{b}_t(\mathbf{g}_j)=0\Leftrightarrow \begin{cases}
\mathbf{a}_{i-1}(\mathbf{g}_j)=0, & t\in \mathbb{F}_q,\\
\mathbf{a}_i(\mathbf{g}_j)=0, & t=\infty.
\end{cases}
$$
Thus $\mathbf{b}_t(\mathbf{g}_j)=0$ if and only if $t=\infty$,
which shows that there is also a  unique $t$ such that $\mathbf{b}_t(\mathbf{g}_j)=0$.
Therefore, a counting argument yields
$$\sum_{t\in \mathbb{P}(\mathbb{F}_q)}\mathbf{1}_{\{\mathbf{b}_t(\mathbf{g}_j)\neq 0\}}=q,
~~~(\mathbf{a}_{i-1}(\mathbf{g}_j), \mathbf{a}_i(\mathbf{g}_j))\neq (0,0).$$

On the other hand, it is easy to see that when $(\mathbf{a}_{i-1}(\mathbf{g}_j),\mathbf{ a}_i(\mathbf{g}_j))= (0,0)$, we have
$$\mathbf{b}_t(\mathbf{g}_j)=0,~ {\hbox{for any}}~ t\in \mathbb{P}(\mathbb{F}_q).$$
This implies that
$$\sum_{t\in \mathbb{P}(\mathbb{F}_q)}\mathbf{1}_{\{\mathbf{b}_t(\mathbf{g}_j)\neq 0\}}=0,
~~~(\mathbf{a}_{i-1}(\mathbf{g}_j),\mathbf{ a}_i(\mathbf{g_j}))= (0,0).$$
In conclusion, the above counting argument yields
$$\sum_{t\in \mathbb{P}(\mathbb{F}_q)}\mathbf{1}_{\{\mathbf{b}_t(\mathbf{g}_j)\neq 0\}}=
\begin{cases}
q, & (\mathbf{a}_{i-1}(\mathbf{g}_j), \mathbf{a}_i(\mathbf{g}_j))\neq (0,0),\\
0, & (\mathbf{a}_{i-1}(\mathbf{g}_j),\mathbf{ a}_i(\mathbf{g}_j))= (0,0).
\end{cases}
$$

Now we consider the sum as follows:
$$\sum_{j\in R_{i-2}}\sum_{t\in \mathbb{P}(\mathbb{F}_q)}\mathbf{1}_{\{\mathbf{b}_t(\mathbf{g}_j)\neq 0\}}=\sum_{\substack{j\in R_{i-2},\\
(\mathbf{a}_{i-1}(\mathbf{g}_j), \mathbf{a}_i(\mathbf{g}_j))\neq (0,0)}}q
=q\big|\{ j\in R_{i-2}\big|(\mathbf{a}_{i-1}(\mathbf{g}_j), \mathbf{a}_i(\mathbf{g}_j))\neq (0,0)\}\big|.$$

By the definition of $S_i$ and $R_i$, we have
\begin{eqnarray*}
&&\big\{ j\in R_{i-2}|(\mathbf{a}_{i-1}(\mathbf{g}_j), \mathbf{a}_i(\mathbf{g}_j))\neq (0,0)\big\}\\[2pt]
&=& \big\{ j\in R_{i-2}|\mathbf{a}_{i-1}(\mathbf{g}_j)\neq 0\big\}\cup \big\{ j\in R_{i-2}|\mathbf{a}_{i-1}(\mathbf{g}_j)=0,
\mathbf{a}_i(\mathbf{g}_j)\neq 0\big\}\\[2pt]
&=&  S_{i-1}\cup \big\{ j\in R_{i-1}|\mathbf{a}_i(\mathbf{g}_j)\neq 0\big\}\\[2pt]
&=& S_{i-1}\cup S_i.
\end{eqnarray*}
Therefore, summing over all $j\in R_{i-2}$ yields the exact averaging identity
\begin{equation}\label{average}
\frac{1}{q+1}\sum_{j\in R_{i-2}}\sum_{t\in \mathbb{P}(\mathbb{F}_q)}\mathbf{1}_{\{\mathbf{b}_t(\mathbf{g}_j)\neq 0\}}
=\frac{q}{q+1}(\big|S_{i-1}\big|+ \big|S_i\big|).
\end{equation}
Fix $i \geq 2$ and by the definition of $S_{i-1}$ we have
\[
\big|S_{i-1}\big| = \max_{\mathbf{a} \in (\mathbb{F}^k_q)^*\setminus\{0\}} \big|\{ j \in R_{i-2} | \mathbf{a}(\mathbf{g}_j) \neq 0 \}\big|.
\]
Hence,
\begin{equation}\label{inequality-2}
\big|S_{i-1}\big|\geq \big|\{j\in R_{i-2}|\mathbf{b}_t(\mathbf{g}_j)\neq 0\big\}|, ~{\hbox{for any}} ~t\in \mathbb{P}(\mathbb{F}_q).
\end{equation}
Therefore, combining Equation (\ref{average}) with (\ref{inequality-2}) yields
$$
\big|S_{i-1}\big|\geq\frac{1}{q+1}\sum_{j\in R_{i-2}}\sum_{t\in \mathbb{P}(\mathbb{F}_q)}\mathbf{1}_{\{\mathbf{b}_t(\mathbf{g}_j)\neq 0\}}=\frac{q}{q+1}\big|(S_{i-1}\big|+ \big|S_i\big|),
$$
which shows that
$$
\big|S_{i-1}\big|\geq q\big|S_i\big|.
$$
Hence, we have
$$\big|S_i\big|\leq \frac{\big|S_{i-1}\big|}{q},$$
which implies the desired result.
\end{proof}

Now we are ready to prove Theorem \ref{mainresult}, which is the main result of this paper.

By construction, $|S_1| = \delta$. Repeated application of the inequality in Lemma \ref{inequality-1} gives:
\[
\begin{aligned}
|S_2| &\leq \left\lfloor \frac{\delta}{q} \right\rfloor, \\
|S_3| &\leq \left\lfloor \frac{|S_2|}{q} \right\rfloor \leq \left\lfloor \frac{\delta}{q^2} \right\rfloor, \\
&\vdots \\
|S_k| &\leq \left\lfloor \frac{\delta}{q^{k-1}} \right\rfloor.
\end{aligned}
\]

Therefore, by Lemma \ref{partition} we obtain that
\[
n = \sum_{i=1}^k |S_i| \leq \delta + \left\lfloor \frac{\delta}{q} \right\rfloor + \cdots + \left\lfloor \frac{\delta}{q^{k-1}} \right\rfloor = \sum_{i=0}^{k-1} \left\lfloor \frac{\delta}{q^i} \right\rfloor.
\]

This completes the proof.

\section{Applications}
As shown in Theorem \ref{mainresult}, Equation (\ref{mainequation}) is stronger than the original antiGriesmer inequality for a projective linear anticode in \cite{CX}
because it holds for all $n$, not just $n<q^{k-1}$
and requires only $d(\mathcal{C}^\bot)\geq 2$.
The rest of this paper derives new consequences of the improved inequality.
\subsection{Some bounds mixing $n,k,\delta$}
By the main result (Theorem \ref{mainresult}), we now can give a lower bound on the diameter $\delta$ of linear anticode.
\begin{Corollary}\label{lowbound1}
Let $\mathcal{C}$ be an $[n,k,\delta]_q$ linear anticode
with  $d(\mathcal{C}^\bot)\geq 2$. Then
\begin{equation}\label{lowbound2}
\delta\geq \left\lceil \frac{n}{1+\frac{1}{q}+\cdots+\frac{1}{q^{k-1}}}\right\rceil.
\end{equation}
That is,
\begin{equation}
\delta\geq \left\lceil \frac{nq^{k-1}(q-1)}{q^{k}-1}\right\rceil.
\end{equation}
\end{Corollary}
\begin{proof}
By Equation (\ref{mainequation}) we have that
$$n\leq \sum_{i=0}^{k-1}\frac{\delta}{q^i}.$$
that is,
$$\delta \geq \frac{n}{1+\frac{1}{q}+\cdots+\frac{1}{q^{k-1}}}.$$
Hence,
$$\delta \geq \left\lceil\frac{n}{1+\frac{1}{q}+\cdots+\frac{1}{q^{k-1}}}\right\rceil
=\left\lceil\frac{n(1-\frac{1}{q})}{(1-\frac{1}{q})(1+\frac{1}{q}+\cdots+\frac{1}{q^{k-1}})}\right\rceil
=\left\lceil \frac{nq^{k-1}(q-1)}{q^{k}-1}\right\rceil.$$
\end{proof}

From the inequality as follows:
$$\frac{nq^{k-1}(q-1)}{q^{k}-1}=\frac{nq^k(1-\frac{1}{q})}{q^k-1}
=\frac{q^k}{q^k-1}\cdot(1-\frac{1}{q})n
>(1-\frac{1}{q})n,
$$
we have
$$\left\lceil\frac{nq^{k-1}(q-1)}{q^{k}-1}\right\rceil>(1-\frac{1}{q})n.$$
This yields the following result, as shown in \cite[Corollary 2.1]{CX}.

\begin{Corollary}(\cite[Corollary 2.1]{CX})
Let $q$ be a prime power and $n$ be a positive integer satisfying $n<q^{k-1}$.
Let $\mathcal{C}\subset \mathbb{F}_q^n$ be a projective linear code of dimension $k$.
Then its maximum weight (diameter) is at least $(1-\frac{1}{q})n$.
\end{Corollary}

Considering the binary case, we can directly deduce the following known lower bound from Corollary \ref{lowbound1}.

\begin{Corollary}(\cite{FARR1973})\label{binarybound}
For a binary linear projective anticode of dimension $k$ and diameter $\delta$,
we have
$$\delta\geq \frac{2^{k-1}n}{2^k-1}.$$
\end{Corollary}

Equation (\ref{lowbound2}) in Corollary \ref{lowbound1} shows that for a fixed length $n$ and dimension $k$,
the diameter $\delta$ cannot be arbitrarily small. The earlier antiGriesmer bound in \cite[Theorem 2.1]{CX}
requires $n<q^{k-1}$ and $d(\mathcal{C}^\perp)\geq 3$, which already implied $\delta \geq (1-\frac{1}{q})n$, see \cite[Corollary 2.1]{CX}.
Our sharper bound removes those restrictions and yields a slightly larger lower bound for small $k$, see the examples as follows.
On the other hand, by comparing Corollary \ref{binarybound} and Corollary \ref{lowbound1},
it is found that our result focuses on more general situations without assuming projectiveness.

\begin{Example}
Let $\mathcal{C}$ be a generalized Reed-Solomon code over $\mathbb{F}_{2^8}$ with parameters $[256, 100, 157]_{2^8}$,
which satisfies the conditions that $n=256<256^{99}=q^{k-1}$ and $d(\mathcal{C}^\bot)=k+1=101\geq 3$.

Simple calculations indicate that
$$\left\lceil\frac{nq^{k-1}(q-1)}{q^{k}-1}\right\rceil=
\left\lceil\frac{256\cdot 256^{100-1}\cdot (256-1)}{256^{100}-1}\right\rceil=256
>\left\lceil(1-\frac{1}{q})n\right\rceil=\left\lceil(1-\frac{1}{256})\cdot256\right\rceil=255.$$
So the lower bound provided by Corollary \ref{lowbound1} is a slightly larger than that of \cite[Corollary 2.1]{CX}.
\end{Example}

\begin{Example}
Let $\mathcal{C}$ be an extended generalized Reed-Solomon code over $\mathbb{F}_{2^8}$ with parameters $[256, 240, 17]_{2^8}$,
and $\mathcal{C}\oplus \mathcal{C}=\{(c_1|c_2)|c_1\in \mathcal{C}, c_2\in \mathcal{C}\}$.
Then $\mathcal{C}\oplus \mathcal{C}$ has parameters $[512, 480, 17]_{2^8}$,
and satisfies that $n=512<256^{479}=q^{k-1}$.

Let $G$ be the generator matrix for $\mathcal{C}$. Then the parity check matrix for $(\mathcal{C}\oplus \mathcal{C})^\perp$ is
$$G'=\begin{pmatrix}
G & 0\\
0 & G
\end{pmatrix}.
$$
Since $d(\mathcal{C}^\perp)=241\geq 3$, any two columns of the matrix $G$ are linearly independent.
Hence any two columns of the matrix $G'$ are linearly independent,
which shows that
$d((\mathcal{C}\oplus \mathcal{C})^\bot)\geq 3$.
Simple calculations indicate that
$$\left\lceil\frac{nq^{k-1}(q-1)}{q^{k}-1}\right\rceil
=\left\lceil\frac{512\cdot 256^{479}\cdot(256-1)}{256^{480}-1}\right\rceil
=511
>\left\lceil(1-\frac{1}{q})n\right\rceil
=\left\lceil(1-\frac{1}{256})\cdot 512\right\rceil=510.$$
So the lower bound provided by Corollary \ref{lowbound1} is a slightly larger than that of \cite[Corollary 2.1]{CX}.
\end{Example}

The following example shows that the lower bound in Corollary \ref{lowbound1} is indeed strictly greater than the known one,
even if it does not meet the code length restriction and the projective property of the dual distance.

\begin{Example}
Let $\mathcal{C}$ be a binary linear code with generator matrix $G=(I_{10}|I_{10})$, where $I_{10}$ is an identity matrix of order 10.
Then $\mathcal{C}$ is a binary $[20,10,2]_2$ linear code.
Simple calculations indicate that
$$\left\lceil\frac{nq^{k-1}(q-1)}{q^{k}-1}\right\rceil=11
>\left\lceil(1-\frac{1}{q})n\right\rceil=10.$$
So the lower bound provided by Corollary \ref{lowbound1} is a slightly larger than that of \cite[Corollary 2.1]{CX}.
\end{Example}

The antiGriesmer bound also can characterize the nonexistence of linear codes with very small diameter.
\begin{Corollary}
Let $\mathcal{C}$ be an $[n,k]$ linear code over $\mathbb{F}_q$ with
 $d(\mathcal{C}^\bot)\geq 2$ and  diameter $\delta$.

(1) If the diameter $\delta$ is less than $n$,
then $\delta$ is bounded  by
$$\delta\geq q.$$
That is, if the diameter $\delta$ is less than $q$,
then $\delta$ is equal to $n$.

(2) If the diameter $\delta$ is less than or equal to $q$,
then the length $n$ is bounded  by
$$n\leq q+1.$$
\end{Corollary}

\begin{proof}
(1) Suppose that $\delta < q$, by Equation (\ref{mainequation}) we have
$$n\leq \delta + \left\lfloor \frac{\delta}{q} \right\rfloor +\cdots +\left\lfloor \frac{\delta}{q^{k-1}} \right\rfloor
= \delta+0+0+\cdots+0=\delta.$$
Since $\delta\leq n$, we get that $\delta=n$. This contradicts the known condition that $\delta<n$. Hence $\delta\geq q$.

(2) With $\delta < q$, we have
$$n\leq \delta + \left\lfloor \frac{\delta}{q} \right\rfloor +\cdots +\left\lfloor \frac{\delta}{q^{k-1}} \right\rfloor
= \delta+0+0+\cdots+0=\delta.$$
So $n=\delta< q$.

With $\delta = q$, we have
$$n\leq \delta + \left\lfloor \frac{\delta}{q} \right\rfloor +\cdots +\left\lfloor \frac{\delta}{q^{k-1}} \right\rfloor
= \delta+1+0+\cdots+0=\delta+1.$$
So $n\leq q+1$.
\end{proof}

In the following we consider the approximate upper bound on the length $n$ and restriction on $k$.
\begin{Corollary}
Let $\mathcal{C}$ be a linear $[n,k]_q$ code over $\mathbb{F}_q$ with
 $d(\mathcal{C}^\bot)\geq 2$ and diameter $\delta$. Then
\begin{equation}\label{upperboundofn}
n\leq
\begin{cases}
\frac{q}{q-1}\delta-1, & \mbox{if}~(q-1)|\delta;\\
\left\lfloor\frac{q}{q-1}\delta\right\rfloor,& \mbox{otherwise}.
\end{cases}
\end{equation}
In particular, for a binary code the inequality $n\leq 2\delta-1$ holds.

Moreover, given fixed $n$ and $\delta$, the dimension $k$ satisfies
\begin{equation}\label{lowerboundofk}
k\geq \log_q(\frac{\delta}{\delta-n(1-\frac{1}{q})}).
\end{equation}
\end{Corollary}

\begin{proof}
Dropping the floor functions in (\ref{mainequation}) only makes the right-hand side larger. Therefore
$$n\leq \sum_{i=0}^{k-1}\left\lfloor \frac{\delta}{q^i}\right\rfloor\leq \sum_{i=0}^{k-1}\frac{\delta}{q^i}<\sum_{i=0}^{\infty}\frac{\delta}{q^i}
=\delta\sum_{i=0}^{\infty}\frac{1}{q^i}=\delta\frac{1}{1-\frac{1}{q}}=\frac{q}{q-1}\delta.
$$
Hence
$$n\leq
\begin{cases}
\frac{q}{q-1}\delta-1, & \mbox{if}~(q-1)|\delta;\\
\left\lfloor\frac{q}{q-1}\delta\right\rfloor,& \mbox{otherwise}.
\end{cases}
$$

As argued above,
$$n\leq \delta(1+\frac{1}{q}+\cdots+\frac{1}{q^{k-1}})=\delta\frac{1-\frac{1}{q^k}}{1-\frac{1}{q}}.$$
Rearranging gives
$$\frac{n}{\delta}(1-\frac{1}{q})\leq 1-\frac{1}{q^k},$$
or
$$q^{-k}=\frac{1}{q^k}\leq 1-\frac{n}{\delta}(1-\frac{1}{q}).$$
Taking logarithms base $q$ yields
$$-k\leq \log_q(1-\frac{n}{\delta}(1-\frac{1}{q})).$$
Solving for $k$ gives
$$k\geq \log_q(\frac{\delta}{\delta-n(1-\frac{1}{q})}).$$
\end{proof}
\begin{Example}
Let $\mathcal{C}$ be the simplex $[2^{k-1}-1, k-1, 2^{k-2}]_2$ code as projective linear anticode with diameter $\delta= 2^{k-2}$.
In that case, $n=2\delta-1$, which makes equal sign of the first inequality of (\ref{upperboundofn}) hold true.

Let $\mathcal{C}$ be $[2q+2, 4, q]_q$ linear code over $\mathbb{F}_q$ with diameter $\delta=2q$ (see \cite[Table 1]{CX}).
In the case when $q>3$, it satisfies
$$\left\lfloor \frac{q}{q-1}\delta\right\rfloor=\left\lfloor \frac{2q^2}{q-1}\right\rfloor
=\left\lfloor \frac{2(q^2-1)+2}{q-1}\right\rfloor=2q+2+\left\lfloor \frac{2}{q-1}\right\rfloor=2q+2=n.$$
which makes equal sign of the second inequality of (\ref{upperboundofn}) hold true.
\end{Example}

\begin{Example}
Let $\mathcal{C}$ be $[\frac{q^k-1}{q-1}, k, q^{k-1}-1]_q$ linear code over $\mathbb{F}_q$
with diameter $\delta=q^{k-1}$ (see \cite[Table 1]{CX}). Simple calculations indicate that
$$\log_q(\frac{\delta}{\delta-n(1-\frac{1}{q})})=\log_q(\frac{q^{k-1}}{q^{k-1}-\frac{q^k-1}{q-1}(1-\frac{1}{q})})
=\log_q(\frac{q^{k-1}}{q^{k-1}-\frac{q^k-1}{q}})=k.
$$
This shows that the equal sign of the inequality of (\ref{lowerboundofk}) hold true.
\end{Example}

\subsection{Some bounds mixing $n,k,w,\delta$}
In this subsection, we establish an upper bound on the length of a linear anticode over $\mathbb{F}_q$
with a prescribed dimension $k$, diameter $\delta$ and  a known $w$-weight.

\begin{Corollary}\label{thm:cf}
Let $\mathcal{C}$ be an $[n,k]_q$ $q$-ary linear code with diameter $\delta$ and $d(\mathcal{C}^\perp)\ge 2$.
If $\mathcal{C}$ has a nonzero codeword of weight $w$, then
\begin{equation}\label{eq:n-ceiling-free}
n \;\le\; w + \frac{q\bigl(1-q^{-(k-1)}\bigr)}{q-1}\,\delta.
\end{equation}
Equivalently,
\[
n \;\le\; w + \Bigl\lfloor \sum_{i=0}^{k-2}\frac{\delta}{q^i}\Bigr\rfloor.
\]
\end{Corollary}

\begin{proof}
Let $\mathrm{Res}(\mathcal{C},c)$ be the residual code after puncturing on the support $I$ of $c$.
Then $\mathrm{Res}(\mathcal{C},c)$ has length $n-w$ and dimension $k-1$.
Since $d(\mathcal{C}^\perp)\ge 2$, we have that $d(\mathrm{Res}(\mathcal{C},c))^\bot\geq 2$.
By Theorem~\ref{mainresult} applied in the residual code $\mathrm{Res}(\mathcal{C},c)$,
$$
n-w\leq \sum_{i=0}^{k-2}\left\lfloor\frac{\delta'}{q^i}\right\rfloor.
$$
Since puncturing cannot increase weight, $\delta'\le \delta$ and so
\[
n-w\leq \sum_{i=0}^{k-2}\left\lfloor\frac{\delta}{q^i}\right\rfloor
\leq \left\lfloor\sum_{i=0}^{k-2}\frac{\delta}{q^i}\right\rfloor
\leq \left\lfloor\frac{q^{k-1}-1}{q^{k-2}(q-1)}\right\rfloor
\leq \frac{q^{k-1}-1}{q^{k-2}(q-1)}\,\delta
= \sum_{i=0}^{k-2}\frac{\delta}{q^i}
= \frac{q\bigl(1-q^{-(k-1)}\bigr)}{q-1}\,\delta,
\]

Rearranging gives
\[
n\;\le\; w+\frac{q^{k-1}-1}{q^{k-2}(q-1)}\,\delta
= w+\sum_{i=0}^{k-2}\frac{\delta}{q^i}
= w+\frac{q\bigl(1-q^{-(k-1)}\bigr)}{q-1}\,\delta,
\]
and~\eqref{eq:n-ceiling-free} follows. This completes the proof.
\end{proof}

Applying Corollary \ref{thm:cf} to a codeword of minimum weight we obtain the following.

\begin{Corollary}\label{cor:w-d-lb}
Under the hypotheses of Corollary~\ref{thm:cf}, we have
\begin{equation}
n \;\le\; d + \frac{q\bigl(1-q^{-(k-1)}\bigr)}{q-1}\,\delta.
\end{equation}
Equivalently,
\[
n \;\le\; d + \Bigl\lfloor \sum_{i=0}^{k-2}\frac{\delta}{q^i}\Bigr\rfloor.
\]
\end{Corollary}

\section*{Data availability}

No data was used for the research described in the article.

\end{document}